\documentclass[letterpaper]{article}
\pdfoutput=1
\usepackage{flairs/flairs}
\usepackage{times}
\usepackage{helvet}
\usepackage{courier}

\usepackage[utf8]{inputenc}
\usepackage{graphicx}
\usepackage{amssymb}
\PassOptionsToPackage{hyphens}{url}\usepackage{hyperref}
\usepackage{url}
\newcommand{\MREE}{MREE}
\newtheorem{definition}{Definition}
\newtheorem{theorem}{Theorem}
\newenvironment{proof}{\strut\\\noindent{\bf Proof~}}{\strut\\\strut\hfill$\bf qed$\\}

\title{Modeling the Feedback of AI Price Estimations on Actual Market Values
\thanks{
On February 15, 2022 we uploaded in overleaf the first draft of this paper under the name {\em Public AI on house price estimations through Zillow may influence a monotonic house price increase and inflation forever according to simulations}, \url{https://www.overleaf.com/read/yttcffkrhvjf\#7120e1}
.}
}

\author{
Viorel Silaghi, Zobaida Alssadi, Ben Mathew, Majed Alotaibi, Ali Alqarni, Marius Silaghi\\Florida Institute of Technology
}
\date{}

\begin{document}
\maketitle

\begin{abstract}
\begin{quote}

Public availability of Artificial Intelligence generated information can change the markets forever, and its factoring into economical dynamics may take economists by surprise, out-dating models and schools of thought.
Real estate hyper-inflation is not a new phenomenon but its consistent and almost monotonous persistence over 12 years, coinciding with prominence of public estimation information from Zillow, a successful Mass Real Estate Estimator (\MREE{}), could not escape unobserved.
What we model is a repetitive theoretical game between the MREE and the home owners, where each player has secret information and expertise.
If the intention is to keep housing affordable and maintain old American lifestyle with broad home-ownership, new challenges are defined.
Simulations show that a simple restriction of MREE-style price estimation availability to opt-in properties may help partially reduce feedback loop by acting on its likely causes, as suggested by experimental simulation models.
The conjecture that the \MREE{} pressure on real estate inflation rate is correlated with the absolute \MREE{} estimation errors, which is logically explainable, is then validated in simulations.
\end{quote}

\end{abstract}

\section{Introduction}

{\em
There is a somewhat dated anecdote~\cite{anecdote2} that goes as follows: In a Native American tribe, the people went to the chief and asked how the winter was expected to be. The chief, having received a modern education, had not been taught how his forefathers had predicted weather decades before. As he was a precautious fellow, he replied ``This will be a cold winter. Go gather wood!'' Desiring to give more accurate advice to his tribe, the chief called the meteorologist to ask how the winter was expected to be. ``This winter will be very cold,'' he replied. The chief returned to his tribe, summoned them, and warned them ''This will be a particularly cold winter. Go and gather more wood!'' Some time later, the chief again called the meteorologist to ask if they were sure the winter would be so cold. The meteorologist replied ``Yes, it seems this winter will be extraordinarily cold!'' The chief summoned his tribe again and warned them ``We are headed in for an extremely rough winter. Go and gather even more wood!'' Ever seeking the most accurate information for his tribe, the chief yet again called the meteorologist a few weeks later to ask if they were sure that the winter would be so cold. ``Yes,'' the meteorologist replied, ``we are more certain than ever that this winter will be the harshest ever seen!'' ``How can you be so sure?'', the chief inquired. ``Never before have we seen the natives so feverishly gathering wood!''
}\\[0.5cm]

Such a phenomenon as the one described in the above anecdote may explain current trends in the American housing market, as observed by realtors~\cite{zillow_distortion}. Artificial Intelligence applications led to Mass Real Estate Estimators (\MREE{}) of which the best well known is Zillow. The MREEs have attained~\cite{zillow_fortune} a prominent position in the estimation of prices for houses in the American market, such that in a seller's market offers to purchase houses at a price lower than MREE's estimation are generally seen as non-starters~\cite{zillow_research}. Thus, most houses end up selling at least for as much as the MREE predicted. The catch is that proprietary predictions like those from a MREE's are very likely influenced by the selling price of neighboring homes. So whenever a house sells above the MREE's estimation, it is expected to increase the predicted value of all nearby houses, maintaining a house price inflation and investment race that further supports a seller's market.

Realtors suggest the estimation errors to 20\% of their classic appraisal approaches~\cite{zillow_under}.
We model a phenomenon based on which, any errors in a widely available MREE's price estimations coupled with the algorithm that feeds them back into selling prices leads to a ratchet effect which we conjecture to explain reasonably well a major factor driving the recently seen explosion of the American housing market. In this report, we model the effects of public availability of price estimates on market behavior and attempt to predict the consequences and potential endgame of the feedback loop created by these estimates, as well as impacts of possible government actions.

\paragraph{Absolute \MREE{} errors propagating into inflation}
The value of a house is a composition of its location value with the value of the construction features.
In a seller's market, when the estimation from \MREE{} for a house is underestimating the house by $\Delta$ to price $x$, the owner or competing bidding buyers will know and transaction the house at the right price $x+\Delta$. At that point the \MREE{} will wrongly conclude that there was an inflation of the location by $\Delta$ and will raise neighboring houses estimates with an impact of the $\Delta$ propagation.

Alternatively, when the estimation from \MREE{} for a given house is overestimating a house by $\Delta$ to a price $x$, the
buyers will be emboldened to trust the \MREE{} reputation and the seller happily adopts the \MREE{} price. This is leading to an immediate actual local inflation by $\Delta$ that propagates in the neighborhood since better houses around will have arguments to be assessed higher.

As such, both overestimating and underestimating errors end up generating inflation phenomena proportional with the absolute error. In the following we will propose models to quantify and simulate these phenomena, and report on experiments with parameters matching publicly available data.

\section{Background and Concepts}

\begin{figure}[!th]
    \centering
    \includegraphics[scale=0.5, trim=100 100 50 20, clip=]{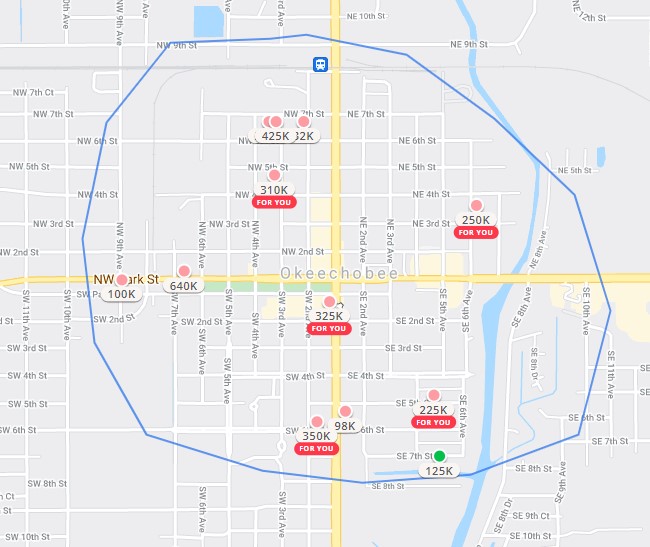}
    \caption{Zillow map with listings (Feb 2022)}
    \label{fig:my_label}
\end{figure}

\begin{figure}[!th]
    \centering
    \includegraphics[scale=0.3, trim=100 100 50 20, clip=]{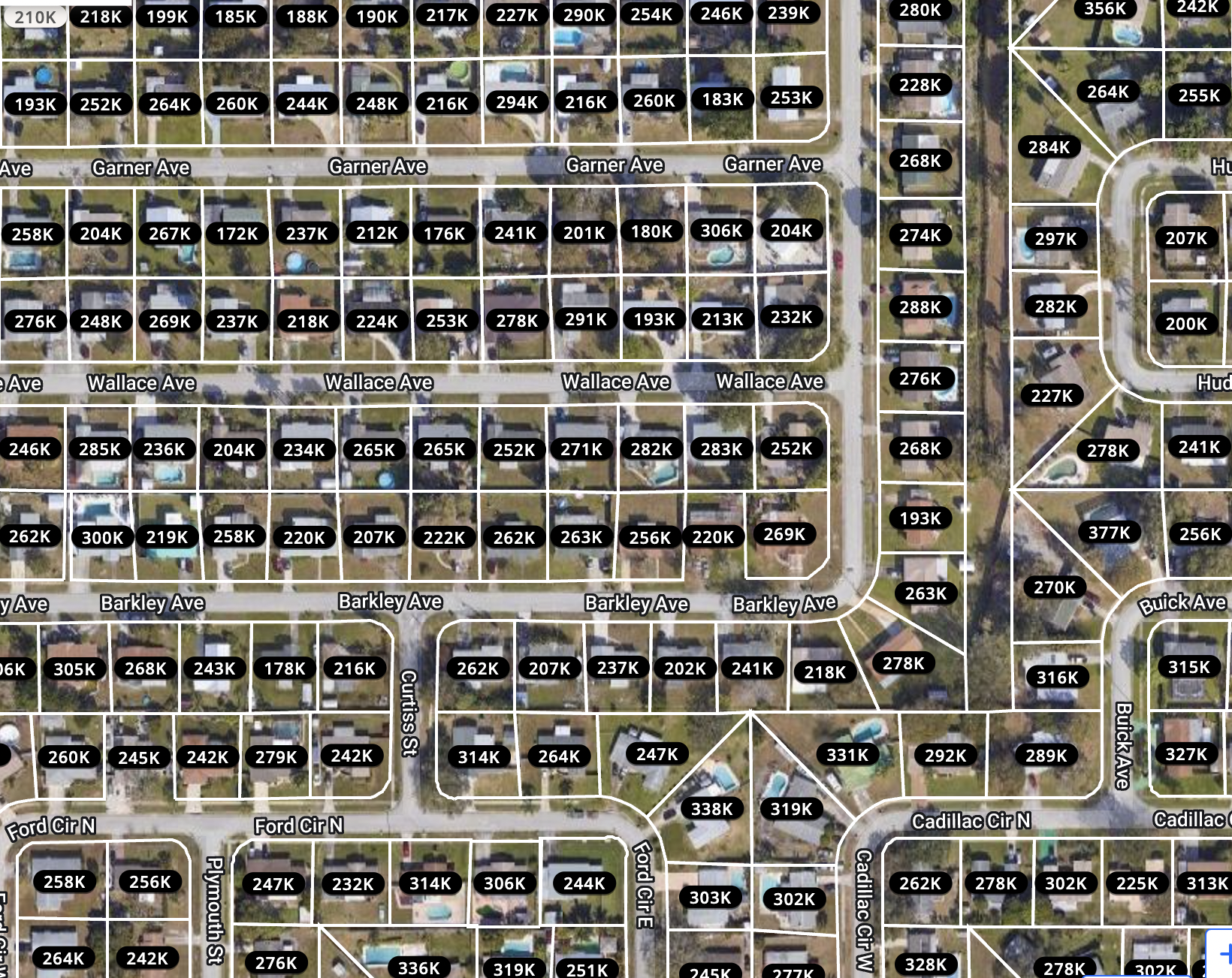}
    \caption{Zillow map with estimations  (Feb 2022)}
    \label{fig:my_label_est}
\end{figure}

In Figure~\ref{fig:my_label}, a snapshot of a Zillow township map is displayed together with listings, while Figure~\ref{fig:my_label_est} shows a Zillow map with single family house estimations. 

Prior research has used the Zillow's Zestimate together with other census bureau data to train deep neural networks for predicting price dynamics~\cite{jiang2021modeling}.

A study and simulation of Zestimate impacts on markets is made available in~\cite{fu2022human},
confirming that sellers tend to sell above the MREE estimations. Their model and simulation assume errors in estimation produce disturbances in opposite directions that would eventually cancel out, and do not address the game theoretic phenomena we raise, where errors in both directions build on each other towards a unidirectional disturbance.
The correlation with experimental data that they note may also be due to the fact that the main studied disturbance there is related to external events, outside the human-algorithm loop, namely in the impacts of the COVID emergency declaration.

Other recent empirical studies of impacts of Zestimate on market outcomes are presented in~\cite{troncoso2023algorithm,zhang2023effect,barnwell2022seller} showing that sellers tend to earn less by listing under the Zestimate, even of they may close faster. It is also shown that listings tend to deviate with amounts that are mainly above the Zestimate.

\section{Case studies of Overestimation and Underestimation Effects}

\begin{table*}
\begin{center}
\begin{tabular}{|ll|}\hline
$\lambda$ & value of home location estimated by MREE\\
$v$ & value of home building features estimated by MREE \\
$u$ & value of home building features estimated by owner \\
$\rho$ & value of market value adjustment estimated by owner \\&(i.e., owner's perceived inflation on the home) \\
$p^{MREE}$& total sell price of home estimated by the MREE\\ 
$p$& closing price of home \\ 
\hline
\end{tabular}
\end{center}
\caption{Notations}\label{table:notations}
\end{table*}

Let us show now with sample cases a realistic model of the world where both overestimation and underestimation errors by \MREE{} lead to house price inflation proportional with the absolute value of the error, even in the absence of external inflation.

Assume that \MREE{} evaluates houses as the sum of two components, a value for location $\lambda$ and a value for the building features, $v$. Also assume that the \MREE{} estimated value for the building features is constant once parsed from available data, assuming there is no inflation outside housing. 

The home owner nevertheless has another estimate for the house features, $u$. The owner also has an estimation of the market value adjustment for the house, in amount of $\rho$. The owner may also have a different estimate for the location but we assume that he can integrate it as correction into $\rho$. These notations are summarized in Table~\ref{table:notations}.

Take two neighboring houses, $A$ and $B$, for whom the \MREE{} estimates the location values to $\lambda_A$ and $\lambda_B$, respectively. Their \MREE{} estimated construction features values are $v_A$ and $v_B$, while the owner and neighbors estimated features values are $u_A$ and $u_B$, respectively.

\paragraph{Overestimation of House A by the \MREE{}}
Assume that $v_A>u_A$, namely specifying that the \MREE{} overestimates $A$. Assume $u_B=v_B$ for a correct estimation thereof.
The \MREE{} total estimated price of the House A is $$p^{\MREE{}}_A=\lambda_A+v_A.$$

The owner of A estimates his home at a lower value of $\lambda_A+u_A+\rho_A=\lambda_A+u_A<p^{\MREE{}}_A$ since $\rho_A=0$ (i.e., no inflation perceived yet).
Based on our assumption that a house will never be sold below the \MREE{} price, since the owner assumes that some buyers will trust the \MREE{},
the owner of House A will list and eventually sell House A at price $p^{\MREE{}}_A=\lambda_A+v_A$.
From the perspective of the \MREE{}, the house was sold as expected and no change will be immediately done in its estimations.

However, the neighbor owning House $B$ who can also estimate himself correctly the value $u_A$ sees the sale and
will infer
that the market values for the house features (or location) have increased, and will increment his own $\rho_B$: 

$$\rho_B \mbox{ is incremented by }  p_A-(\lambda_A+u_A).$$

Next time when B is sold, its owner and the bidders will transaction House B at price $p_B=\lambda_B+u_B+\rho_B$ which based on our aforementioned assumption that 
$$v_B=u_B,$$ 
is larger than \MREE{}'s estimation 
$p^{\MREE{}}_B=\lambda_B+v_B$ with the increment $\rho_B$.

$$ \lambda_B+u_B+\rho_B=p_B > p^{\MREE{}}_B= \lambda_B+v_B $$

$$ p_B = p^{\MREE{}}_B  + \rho_B$$

At that point, the \MREE{} takes note of the unexpected price in the selling of B (and other similar neighboring house prices), and infers that the location value
has increased to $\lambda'_B=p_B-v_B=\rho_B$.

{\it Conclusion Overestimation Case.} Therefore an overestimation of house feature values by the \MREE{} leads to a corresponding inflation in the estimation of the location, and therefore to an inflation of all surrounding house prices. It is worth noticing that this inflation will repeat each time House A is sold followed by the sale of another house in its neighborhood, since the \MREE{} will each time be surprised by the effect of the values $\rho$ and $u$ which are hidden for it.

\paragraph{Underestimation of House A by the \MREE{}}
Using the same situation of two neighboring houses let us assume
that $v_A<u_A$, namely specifying that the \MREE{} underestimates House A's construction features to $v_A$ instead of $u_A$. At the beginning $\rho_A=\rho_B=0$. The owner of House A will sell his house at the price estimated by himself, namely.

$$p_A = \lambda_A+u_A+\rho_A = \lambda_A+u_A$$

which is larger that the price estimated by the \MREE{}:

$$p_A^{\MREE{}} =\lambda_A+v_A $$

The owner of House A can sell at the real price because some of the bidders visit the home and will understand the correct value.

The \MREE{} takes note of the sale price and erroneously infers that the price of the location has increased by difference of:

$$ p_A - p_A^{\MREE{}} = u_A-v_A $$

This will increase the location estimation for all neighboring houses, including House B to $\lambda'_B$ and House A to $\lambda'_A$. When House B is sold of the new price inflated from this process, the new owner of House A will take note of the new prices for location, as an average of house estimations in the neighborhood after subtracting construction values, and will again be able to sell House A at
$\lambda'_A+u_A$, repeatedly generating inflation.

{\it Conclusion Underestimation Case.} Therefore, an underestimation of house feature values by the \MREE{} leads to a corresponding inflation in the estimation of the location, and therefore to an inflation of all surrounding house prices. 
It is worth noticing that also this inflation from underestimation will repeat each time House A is sold followed by the sale of another house in its neighborhood, since the \MREE{} will each time be induced to inflate the location price estimate.

\paragraph{Assumptions that conditioned the repeated inflation}
Here we will take note of the assumptions we had made and that led to the situation where both underestimation and overestimation errors from \MREE{} lead to repeated inflation in time even in the absence of outside inflation.
\begin{enumerate}
\item
No other inflation occurs outside the one generated by the housing market.
    \item
All sales are made at the maximum between the estimation made by the \MREE{} and the real estimation made by the house owner and some bidders, during a persistent seller market.
    \item
The \MREE{} estimates house prices as sums between the location value estimation $\lambda$ which evolves over time and the value of the construction features that is extracted from the house description and does not evolve in the absence of external inflation. 
\item
The location value propagates in neighborhood, and can be estimated by house owners as average value of \MREE{} estimations in the neighborhood minus the corresponding houses construction feature values that they know.
\item
Owners assume that a house construction feature value inflation has occurred when houses in their neighborhood are sold above the price given as sum between its real construction features value and location value.
\end{enumerate}

We note that the above assumptions can be reasonable estimations of the current housing market, and that the inflation that the model generates can be a support for the prolonged seller market with buyer exuberance as assumed~\cite{zillow_exuberance}.

\section{Model Integrating Information}

We formalize the Real Estate Prediction Problem (REPP) with the simplifying assumption that it operates under a sustainable seller's market environment and therefore it is reasonable to assume that listing and closing times are independent of market prices and evolution. The REPP problem introduces the assumption that a public information function $P$ is provided by a \MREE{} for the estimation of prices at current time.

\begin{definition}[REPP]
A Real Estate Prediction Problem (REPP) of a \MREE{} is defined by a tuple $\langle G,T,P,\Lambda\rangle$ where G is a graph $G(N,E)$ consisting of a set of nodes N representing houses and arcs E representing distances between the houses. 
An edge $e\in E$ corresponds to a distance between the nodes.

Each node $n\in N$ is a tuple $\langle v,u,\lambda^{n,0},\rho^{n,0}\rangle$ where $v$ is the value of the house from the perspective of its construction features as estimated by \MREE{} and $u$ is the objective actual construction features value estimated by owners assumed to be experts, while $\lambda^{n,0}$ is the value of the location at the initial time $t_0$, and $\rho^{n,0}$ is a market value adjustment of the construction features at the initial time.

$T$ is a vector of transactions, each $k^{th}$ transaction $\theta\in T$ being a tuple $\langle t_\theta,c_\theta,i_\theta\rangle$ where $t_\theta$ is the day of the transaction contract, $c_\theta$ is the day of transaction contract closing, and $i_\theta$ is the transacted node. The vector $T$ is ordered by closing time.

For each node $n$, $\lambda^{n,k}$ and $\rho^{n,k}$ are variables specifying the location and the construction market value estimations after the $k^{th}$ transaction closing, respectively.
For a given $k$, the set of variables $\lambda^{n,k}$ is denoted $\lambda^k$ and the set of variables $\rho^{n,k}$ is denoted $\rho^k$.

The price $p_\theta$ of each transaction $\theta$ is given by a function $P(G,\theta, \lambda^{c(\theta^k)}, \rho^{c(\theta^k)})$ where the $c(\theta^k)^{th}$ transaction is the last one closing before the contract date $t_{\theta^k}$ of $\theta^k$.

Each closing of $k^{th}$ transaction $\theta^k$ with price $p_{\theta^k}$ has an impact on the estimations of the location value of each node $n$, $\lambda^{n,k}$, given by a function $\Lambda(n,G,\theta^k,p_{\theta^k}, \lambda^{k-1})$, which cannot access the components $u$ of nodes in $G$. 
Also, after $\theta^k$ the construction feature value $\rho^{n,k}$ are given by a function $\Pi(n,G,\theta^k,p_{\theta^k}, \lambda^{k-1}, \rho^{k-1})$.

The REPP problem is to compute the estimates after the last transaction.
\end{definition}

\begin{theorem}
Assuming $P$ and $\Lambda$ are polynomial, the REPP problem can be solved in polynomial time.
\end{theorem}

\begin{proof}
The functions $P, \Pi$, and $\Lambda$ only have to be applied once for each transaction and each house in $G$.
\end{proof}

Using solutions to REPP problems one can estimate the real estate inflation as a function of the error \MREE{} has in estimating $v$ instead of $u$.

\paragraph{Studied REPP functions}

In our reported research, we study functions:
\begin{itemize}
    \item 
$P$ defined as:

$p_{\theta^k} = P(G, \theta^k, \lambda^{c(\theta^k)}, \rho^{c(\theta^k)}) = \lambda^{i_{\theta^k},c(\theta^k)}+\max(v_{i_{\theta^k}},u_{i_{\theta^k}}+\rho^{i_{\theta^k},c(\theta^k)})$.
\item
$\Lambda$ defined as:

\noindent
\begin{eqnarray}
\lambda^{n,k} &=& \Lambda(n,G,\theta^k,p_{\theta^k}, \lambda^{k-1})= \\
&=& \label{eq:lambda-update}
\lambda^{n,k-1}+\left(\frac{p_{\theta^k}-b*v_{i_{\theta^k}}}{ 1 + a * v_{i_{\theta^k}} } -\lambda^{i_{\theta^k},k-1}\right)
\\
&&~~~~~~~~~~~~~~*ReLU\left(\frac{R-d(n,i_{\theta^k})}{R}\right)\nonumber
\end{eqnarray}

with b=1 and a=0, while R is a maximum influence distance and $d(n,m)$ is the Euclidean distance between nodes $n$ and $m$. In experiments with grid maps, $\frac{R-d(n,i_{\theta^k})}{R}$ is replaced with $\frac{R_x-d_x(n,i_{\theta^k})}{R_x}\frac{R_y-d_y(n,i_{\theta^k})}{R_y}$
where $R_x$ and $R_y$ are maximum influence distances on $x$ and $y$ coordinates, respectively, and $d_x()$ and $d_y()$ are Euclidean distance functions along projections on the $x$ and $y$ coordinates.
ReLU(x) is the Rectified Linear Unit function returning its parameter if it is positive and 0 otherwise.
\item
Function $\Pi$:
\noindent
$\rho^{n,k} = \Pi(n,G,\theta^k,p_{\theta^k}, \lambda^{k-1}, \rho^{k-1})=$\\
\hspace*{7mm}$=
\rho^{n,k-1}+ReLU\left(\frac{p_{\theta^k}-MREE_{\theta^k}}{a * \lambda^{\i_{\theta^k},k-1} + b}-u_{i_{\theta^k}}\right)$* $
ReLU\left(\frac{R-d(n,i_{\theta^k})}{R}\right)
$.

Where
$MREE_{\theta^k} = \lambda^{\i_{\theta^k},k-1}+v_{\i_{\theta^k},k-1}$
is the MREE public estimation of the value of the house in the moment of the closing. 
\end{itemize}

\section{Model With Information Opt-In}

Here we start from the assumption that the society would prefer to not have a house inflation disconnected from general inflation, and to avoid increasingly long streaks of strong house market bubbles. However, artificial intelligence is bound to continue to offer slightly erroneously information that sellers and buyers nevertheless need and the persistent existence of \MREE{} seems hard to avoid, which according to our model is bound to exacerbate such undesired phenomena. The question is how can these side effects be mitigated.

Rather than analyzing heavy governmental policies, we here propose to analyze the impact of a common consumer freedom approach involving protection of privacy based restriction of information listing by \MREE{}s to those real estate properties that opt-in.

We further assume that properties that are not listed and made handily available to market participants, will therefore not generate neighbor resentment and correction of the type captured by our model in the variables $\rho$.

We will assume that houses not opting-in still have access to their own \MREE{} estimation at the time of transactions.
With a large percentage of the houses not opting in the \MREE{} information system, 
we assume that the lack of information leads to not updating the $\rho$ variables on transaction with houses not opted-in the system, while their transaction is based on user utility and MREE location estimation:

$p_{\theta^k} = P(G, \theta^k, \lambda^{c(\theta^k)}, \rho^{c(\theta^k)}) = \lambda^{i_{\theta^k},c(\theta^k)}+u_{i_{\theta^k}}+\rho^{i_{\theta^k},c(\theta^k)}$.

This modification conducts to reduction of error propagation into inflation due to construction value overestimation. In our model where the errors of the published \MREE{} price estimation is equally spread into overestimation and underestimation, the expectation is that at most half of the error impacts could be reduced, the errors due to overestimation. This expectation is confirmed by simulations reported in the next section. The impact can be increased when the opt-in based \MREE{} is biased more to overestimation than underestimation at the same absolute error range.

\section{Simulations Discussion}

We use simulations to verify the consistency of our model  and to validate the theoretically inferred expectations described in prior sections.

We simulate a square city of 1001 by 1001 equally distanced houses.
Firstly, key location points are selected in the nodes of a grid with cell sides given by a parameter NEIGHBORHOOD-SIZE sampled between 5 and 25 in our experiments. A location value $\lambda$ is given to each of these key points, selected randomly in the range $\lambda_{min}=\$30000$ to \$230000. The initial location value of each intermediary house is computed with bi-linear interpolation between the location values of these grid points.

Further we initialize the estimation of the construction feature valuations $v$ for each house to a value $x$ set randomly between \$100000 and \$600000, scaled by multiplication with a factor telling how much more expensive the location is compared to the cheapest location, thereby simulating the fact that expensive locations tend to have more expensive construction features.

$$v=x*(1+\frac{\lambda-\lambda_{min}}{\lambda_{min}*10})$$

The error-free construction feature valuations is set to $v\pm{}\frac{\epsilon}{2}$, where $\epsilon$ is the \MREE{} error
randomly generated in a range given by the simulation instance, between \$5000 and \$50000.

Based on an assumption that in average each house is resold every 2000 days (approx. 5.5 years),
every day a fraction of 0.05\% of the houses are listed for sale, and a listing is available in average for 5 days before a sale offer is received. This is implemented by starting offers of the same number of houses as the number of new listings, 5 days after the beginning. A transaction happens 30 days after the offer, and it is at the moment of the transaction that $\lambda$ and $\rho$ values are updated according to the aforementioned model.

\paragraph{Simulation I: Inflation based on estimated price.}

\begin{figure}[!th]
    \centering
    \includegraphics[scale=0.6, trim=0 0 0 0, clip=]{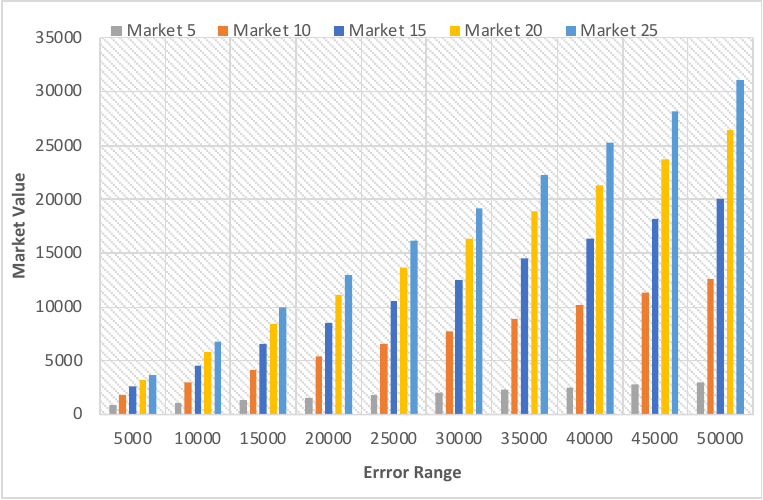}
    \caption{Market inflation as function of absolute error range, various neighborhood sizes: 5 houses to 25 houses}
    \label{fig:market_by_error}
\end{figure}

In a first reported experiment we tested the evolution of the simulated housing market as a function of the absolute error range in which the \MREE{} estimates house values.  The range was sampled between \$5000 and \$50000 with increments of \$5000. The experiment was run separately for markets where the $R$ NEIGHBORHOOD-SIZE parameter in the \MREE{} location update equation~\ref{eq:lambda-update} is set to either 5, 10, 15, 20, or 25 houses.
The results are shown in Figure~\ref{fig:market_by_error}. They show a linear relation with a slope that is proportional with both the absolute error range and the neighborhood size, as expected from the theoretic model. The largest gain in inflation occurs when the neighborhood size increases from 5 to 10 houses distance.

\begin{figure}[!th]
    \centering
    \includegraphics[scale=0.6, trim=0 0 0 0, clip=]{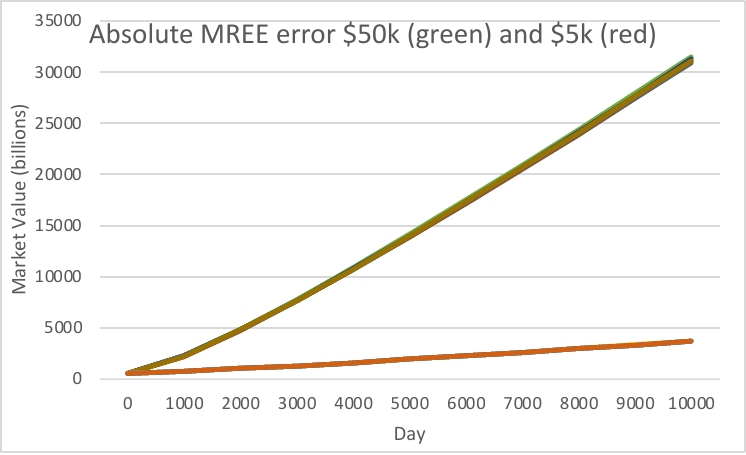}
    \caption{Market inflation as function of time in days for the two extreme absolute error ranges considered: \$5k houses and \$50k houses (green) at neighborhood distance 25}
    \label{fig:market_by_day_house_price}
\end{figure}

In another set of experiments depicted in Figure~\ref{fig:market_by_day_house_price} we verify the impact of relative house construction value on inflation. 
The inflation is shown to be driven faster in the presence of expensive houses. This simulation was performed with neighborhood size of 25.

\begin{figure}[!th]
    \centering
    \includegraphics[scale=0.6, trim=0 0 0 0, clip=]{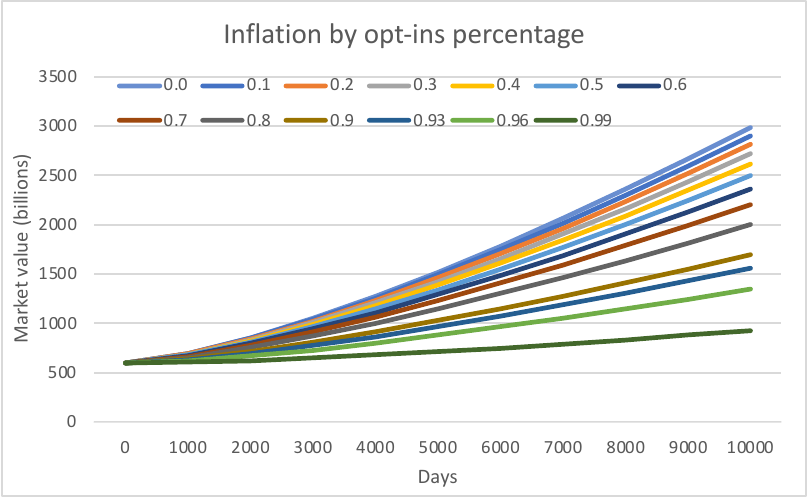}
    \caption{Market inflation as function of time at different percentages of opt-in for listings absolute error range \$10k 
     at neighborhood distance~10. The numbers next to collors specify the  percentage of homes who do not opt in.}
    \label{fig:market_by_day_optin}
\end{figure}
\subsection{Simulation II: Inflation with opt-in MREE listings.}

Another set of experiments verifies the impact of an opt-in policy, and the results are summarized in Figure~\ref{fig:market_by_day_optin}. Each curve is associated with a percentage of people who did not opt in.
It can be observed that when everyone opts-in the MREE (the blue line) the inflation is maximized while with 1\% opt-in (dark curve), the observed inflation was significantly reduced.

Our experiments do not factor in other external constraints, like the lack of money in the system that can appear at various moments and put pressure on reducing or even temporarily freezing the inflationary trend. However, such effects would 
occur at other price equilibrium that is different from the one in the absence of the MREE AI-human feedback loop.

\section{Conclusions}
The increased presence of AI products in society has raised concerned about the recursive impact on itself and unexpected consequences in society. A recurrent recent 
worry pertains to the possibility that outputs from LLMs will be used to taint training data for future models.

In the same line of thought, but in a very different application of Artificial Intelligence, we model ways
in which price prediction models employed by Mass Real Estate Estimators (MREEs), like Zillow, can have recursive effects through society to future estimates, yielding new phenomena in the real estate markets, and we look for models of impacts that explain recent unusual markets behaviors. What we obtain is a repetitive theoretical game between the MREE and the home owners, where each player has secret information and expertise. In particular, we build on the observation that due to game theoretic effects, both overestimation and underestimation errors from MREE yield increasing inflationary pressure, while inflation dampening effects would mainly come from external factors like the lack of cash liquidity in the market, world events, regulations, and significant oversupply.

A formal parametric model is developed that explain how both underestimation and overestimation errors in MREE evaluations of individual home construction features value can have recursive ripple effects into global inflation.
A simulator is built for a town with a grid of homes and realistic distribution of location values and home construction features, where the MREE produces estimations that are rational given its information and according to the formal model proposed, and where the home owners react rationally given their knowledge.

We also simulate the effect of possible government control policies. In particular we find that an opt-in requirement for allowing publications of listings by MREE can have a
significant softening impact on the amplifying factor of the inflationary feedback loop created by inevitable estimation approximation errors.
While the impact of external disasters and lack of liquidity may increase or dampen inflationary pressure, price equilibrium would be achieved at other values than in the absence of the MREE AI-human feedback loop.

\bibliographystyle{flairs/flairs}
\bibliography{main.bib}

\end{document}